\documentclass[pra,twocolumn,superscriptaddress,tightenlines,longbibliography]{revtex4-1}

\usepackage{color}
\usepackage[pdftex]{hyperref}
\hypersetup{
    unicode=false,          
    pdftoolbar=true,        
    pdfmenubar=true,        
    pdffitwindow=false,      
    pdfnewwindow=true,      
    colorlinks=true,       
    linkcolor=blue,          
    citecolor=blue,        
    filecolor=blue,      
    urlcolor=blue,          
}

\usepackage{amsmath,amsthm,amsopn,amsfonts,amssymb,mathtools,nicefrac}
\usepackage[capitalize]{cleveref}
\usepackage{mathrsfs}
\usepackage{comment}
\usepackage{braket}
\usepackage{thm-restate}
\usepackage{thmtools}

\usepackage{tikz}
\usetikzlibrary{positioning}

\DeclareMathOperator{\rank}{rank}
\DeclareMathOperator{\Tr}{Tr}

\DeclareMathOperator{\spn}{span}
\usepackage[normalem]{ulem}

\DeclarePairedDelimiter{\norm}{\lVert}{\rVert}

\declaretheorem[numberwithin=section]{theorem}
\declaretheorem[sibling=theorem]{lemma}
\declaretheorem[sibling=theorem]{corollary}
\declaretheorem[sibling=theorem,name=Proposition]{proposition}

\declaretheorem[sibling=theorem]{remark}

\declaretheorem[numbered=no,style=definition,name=Fact]{fact*}

\newtheorem*{theorem*}{Theorem}
\newtheorem*{corollary*}{Corollary}
\newtheorem*{proposition*}{Proposition}
\newtheorem*{lemma*}{Lemma}
\newtheorem*{claim*}{Claim}
\newtheorem*{problem*}{Problem}

\newcommand{\Hilbert}{\mathcal{H}}

\newcommand{\underflow}[2]{\underset{\kern-60mm \overbrace{#1} \kern-60mm}{#2}}

\newcommand{\CZ}{U}
\newcommand{\CZZ}{CZZ}

\newcommand{\louis}[1]{{\color{green}{louis: #1}}}
\newcommand{\es}[1]{{\color{red}{Edgar: #1}}}
\newcommand{\esnew}[1]{{\color{blue}{Edgar New Text: #1}}}

\newcommand{\eric}[1]{{\color{red}{EC: #1}}}

\newcommand{\cgme}[1]{\mathcal{C}_{\text{GME}}(#1)}
\newcommand{\ngme}[1]{\mathcal{N}_{\text{GME}}(#1)}
\newcommand{\ggme}[1]{\mathcal{G}_{\text{GME}}(#1)}
\newcommand{\dya}[1]{\ket{#1}\bra{#1}}
\newcommand{\R}[1]{\ket{R_{#1}}}
\newcommand{\Rn}{\ket{R_{2n+1}}}
\renewcommand{\L}[1]{\ket{L_{#1}}}

\newcommand{\SLOCC}[2]{\ket{#1}\xrightarrow[]{{\scriptscriptstyle SLOCC}}\ket{#2}}

\newcommand{\mbf}{\mathbf}
\newcommand{\mbb}{\mathbb}

\newcommand{\op}[2]{|#1\rangle\langle #2|}
\newcommand{\ip}[2]{\langle #1| #2 \rangle}

\begin{document}
\title{
Tensor Rank and Other Multipartite Entanglement Measures of Graph States
}
\author{Louis Schatzki}
\email{louisms2@illinois.edu}
\affiliation{Department of Electrical and Computer Engineering, University of Illinois at Urbana-Champaign}
\affiliation{Illinois Quantum Information Science and Technology (IQUIST) Center, University of Illinois Urbana-Champaign}

\author{Linjian Ma}
\affiliation{Department of Computer Science, University of Illinois at Urbana-Champaign}
\affiliation{Illinois Quantum Information Science and Technology (IQUIST) Center, University of Illinois Urbana-Champaign}

\author{Edgar Solomonik}
\affiliation{Department of Computer Science, University of Illinois at Urbana-Champaign}
\affiliation{Illinois Quantum Information Science and Technology (IQUIST) Center, University of Illinois Urbana-Champaign}

\author{Eric Chitambar}
\email{louisms2@illinois.edu}
\affiliation{Department of Electrical and Computer Engineering, University of Illinois at Urbana-Champaign}
\affiliation{Illinois Quantum Information Science and Technology (IQUIST) Center, University of Illinois Urbana-Champaign}

\begin{abstract}
    Graph states play an important role in quantum information theory through their connection to measurement-based computing and error correction. Prior work has revealed elegant connections between the graph structure of these states and their multipartite entanglement content.  We continue this line of investigation by identifying additional entanglement properties for certain types of graph states.  From the perspective of tensor theory, we tighten both upper and lower bounds on the tensor rank of odd ring states ($\Rn$) to read  $2^{n}+1\leq\rank(\Rn)\leq 3\cdot 2^{n-1}$.  Next, we show that several multipartite extensions of bipartite entanglement measures are dichotomous for graph states based on the connectivity of the corresponding graph. Lastly, we give a simple graph rule for computing the $n$-tangle $\tau_n$.
\end{abstract}
\maketitle

\section{Introduction}
Entanglement is one of the defining properties of quantum systems \cite{Einstein1935} and has been recognized as a fundamental resource for quantum information processing \cite{Horodecki2009,Ekert1998}. A pure state is considered to be entangled if it cannot be written in the form $\ket{\psi} = \bigotimes_{i=1}^n\ket{\psi^{(i)}}$. Similarly, a mixed state is entangled if it cannot be written as $\rho = \sum_{k}p_k\bigotimes_{i=1}^n \rho_k^{(i)}$.

Quantifying the amount of entanglement in a quantum state is not always straightforward. For pure bipartite systems, the Schmidt decomposition and resulting spectra fully characterize the entanglement properties and transformations under local operations and classical communication (LOCC) \cite{Horodecki2009,Ekert_schmidt,Walter2016}. The Schmidt decomposition of a pure state takes the form $\ket{\psi} = \sum_i \sqrt{\mu_i}\ket{u_i}\ket{v_i}$, where $\mathbf{\mu}=\{\mu_i\}$ are the Schmidt coefficients and $\{\ket{u_i}\}$ and $\{\ket{v_i}\}$ are sets of orthogonal states. Further, a variety of entanglement measures are known, i.e. functionals $E(\rho)$ that are non-increasing (on average) under LOCC and $E(\rho) = 0 $ if $\rho$ is a separable state \cite{Vidal_entmono}. Examples include the entanglement of formation \cite{Bennett_Ef}, distillable entanglement \cite{Rains_Ed},  negativity \cite{Vidal_neg, Lee_ms_neg}, geometric measure \cite{Wei_geomeas}, and concurrence \cite{Wootters_conc,Mintert_conc}. However, the picture grows significantly more complicated when considering multipartite entanglement, as we discuss below.

In this work we consider the amount and form of multipartite entanglement that arises in a class of quantum states known as graph states. These are of particular interest due to their application in measurement based quantum computing \cite{Raussendorf2001,Hein_app_ent_gs}, error correction codes \cite{Schlingemann_graph_EC}, secret sharing \cite{Markham_gs_secret}, and stabilizer computation simulation \cite{Anders_gs_simulation}. Further, by studying the entanglement properties of graph states, we actually quantify the entanglement of the larger set of stabilizer states. This follows from the fact that every stabilizer state is equivalent under local unitaries to at least one graph state \cite{Nest2004_loc_comp, Schlingemann_stab_graph}. As entanglement measures are invariant under local unitaries, one thus need only consider graph states to analyze all stabilizer states.

In this work we focus on the tensor rank and the GME concurrence, negativity, and geometric measures of entanglement.  Our main contributions to the study of graph states are twofold.  First, we consider ring states of $2n+1$ qubits and sharpen the bound on tensor rank from $2^{n}\leq\rank(\Rn)\leq 2^{n+1}$  to $2^{n}+1\leq\rank(\Rn)\leq 3\cdot 2^{n-1}$.  While this may seem like incremental progress, we stress that computing the tensor rank is a very challenging problem, and any progress in this direction is noteworthy.  Indeed, the analysis we employ goes beyond the bipartite bounding techniques of previous approaches.  This work thus contributes to the steadily growing research on the tensor rank of multipartite entangled states \cite{eisert2001schmidt, Chitambar_rank, Chen2010_entcost, Yu-2014a, Vrana-2015a, Vrana-2017a, Chen-2017a, Christandl-2018a, Bruzda-2019a}. Operationally, the improved bounds help better characterize the amount of entanglement needed to generate ring states using LOCC. Second, we study the GME concurrence, negativity, and geometric measure for general graph states.  These are shown to be sharply dichotomous and having a constant value for all connected graphs.  Before presenting our results, we briefly review the main concepts considered in this paper.

\subsection{Schmidt Measure and Tensor Decomposition}

Any $N$-party pure state $\ket{\psi} \in \Hilbert^{(1)} \otimes \cdots \otimes \Hilbert^{(N)}$  can be represented as 
\begin{equation}\label{eq:cpd}
    \ket{\psi} = \sum_{i=1}^R \mu_i\ket{\psi_i^{(1)}}\otimes \cdots \otimes \ket{\psi_i^{(N)}},
\end{equation}
where each $\ket{\psi_i^{(j)}}\in \Hilbert^{(j)}$. 
When $\ket{\psi}$ is viewed as an $N$-dimensional tensor, \cref{eq:cpd} is also known as a canonical polyadic (CP) tensor decomposition~\cite{hitchcock1927expression,harshman1970foundations}.
 The CP rank $r = \rank (\ket{\psi})$ of a tensor is defined as the smallest $R$ such that \eqref{eq:cpd} can be satisfied.  The CP rank is also known as the tensor rank, and we will use both types of terminology throughout this paper. In general, finding the CP rank of a tensor is NP-hard~\cite{haastad1989tensor}.

Different from the matrix case, for tensors the best rank-$R$ approximation may not exist. And there exists tensors that can be approximated arbitrarily well by rank-$R$ tensors where $R<\rank(\ket{\psi})$.  
In this case, \textit{border rank}~\cite{bini2007role,bini1979n2} is defined as the minimum number of rank-one tensors that are sufficient to approximate the given tensor with arbitrarily
small error.

The tensor rank is a \textit{bona fide} entanglement measure \cite{eisert2001schmidt} that is particularly useful studying state transformations under stochastic local operations and classical communication (SLOCC). These are transformations such that $\SLOCC{\psi}{\phi}$ with some non-zero probability (and is thus a generalization of LOCC). It is known that if $\SLOCC{\psi}{\phi}$, then $\rank{\ket{\psi}} \geq \rank{\ket{\phi}}$ \cite{Chitambar_rank}. Note that we can characterize SLOCC equivalence (i.e. $\SLOCC{\psi}{\phi}$ and $\SLOCC{\phi}{\psi}$) via invertible operators:
\begin{equation}\label{eq:slocc_equiv}
    \ket{\psi} = A_1\otimes A_2 \otimes \ldots \otimes A_n \ket{\phi},
\end{equation}
implying that $\rank (\ket{\psi}) = \rank (\ket{\phi})$. Lastly, we note that the tensor rank relates to entanglement cost. In particular, a generalized d-dimensional GHZ state (or any equivalent state) can be converted to an arbitrary state $\ket{\psi}$ iff $d\geq \rank(\ket{\psi})$ via SLOCC \cite{Chen2010_entcost}.  This provides an operational meaning to the tensor rank in terms of the entanglement resources needed to build $\ket{\psi}$ using GHZ states in the distributed setting.

\subsection{Measures of Genuine Multipartite Entanglement}

An $N$-partite pure state $\ket{\psi}$ is said to have genuine multipartite entanglement (GME) if it is not a product state under any bipartition $A|\overline{A}$ \cite{Vicente_gme,Dai_gme}; i.e. $\ket{\psi}\not=\ket{\alpha}\otimes\ket{\beta}$, where $\ket{\alpha}$ is held by parties in $A$ and $\ket{\beta}$ is held by parties $\overline{A}$.  States for which $\ket{\psi}=\ket{\alpha}\otimes\ket{\beta}$ are called biseparable, and in general it may be desirable for a multipartite entanglement measure to capture how close a state is to being biseparable. Accordingly, one can define measures via minimization over all possible bipartitions. That is, given some bipartite entanglement measure $E(\ket{\phi}),\ \ket{\phi}\in\mathcal{H}_A\otimes\mathcal{H}_B$, we take the multipartite extension to be $\min_{A}E_A(\ket{\psi})$, where $E_A$ is the measure $E$ evaluated according to partition $A|\overline{A}$.  Note that $E(\ket{\psi})=0$ if $\ket{\psi}$ is biseparable according to some partition. Thus, this multipartite extension is faithful with respect to GME. In this work we consider, beyond tensor rank, the following:
\begin{itemize}
    \item GME concurrence~\cite{Ma_gme_conc},
\begin{equation}\label{eq:gme_conc}
    \cgme{\rho} = \min_{A}\sqrt{2(1-\Tr[\rho_A^2])},
\end{equation}
\item GME negativity~\cite{Vidal_neg},
\begin{equation}\label{eq:gme_neg}
    \ngme{\rho} = \min_{A}\frac{1}{2}\left(\norm{\rho^{T_A}}_1-1\right),
\end{equation}
\item and GME geometric measure~\cite{Dai_gme},
\begin{equation}\label{eq:gme_geo}
    \ggme{\rho} = \min_{A}\left(1 -\max_i\mu_i\right),
\end{equation}
\end{itemize}
where $\norm{\cdot}_1$ denotes the Schatten 1-norm and $\rho^{T_A}=\mathbb{I}_{\overline{A}}\otimes T_A(\rho)$ is the partial transpose, and $\mu_i$ are the Schmidt coefficients from the Schmidt decomposition according to partition $A|\overline{A}$. The geometric measure for bipartite systems takes this form, but this is different than the general definition \cite{Wei_geomeas}.  


Lastly, we also evaluate the $n$-tangle $\tau_n$ \cite{Coffman2000_tangle,Wong2001_ntangle} on graph states. For pure states of even numbers of qubits, this is defined as 
\begin{align}\label{eq:n_tangle}
    \tau_n(\ket{\psi}) = \lvert\braket{\psi | \tilde{\psi}}\rvert^2,
\end{align}
where $\ket{\tilde{\psi}} = \sigma_y^{\otimes n}\ket{\psi^*}$ and $\ket{\psi^*}$ indicates the complex conjugate. 
The $n$-tangle is the square of a quadratic SLOCC invariant and can thus be used to distinguish between types of multipartite entangled states \cite{Li_nt_slocc}.

\subsection{Graph States}
Graph states are quantum states corresponding to some graph $G=(V,E)$, where $V$ is the vertex set and $E$ is the edge set with corresponding adjacency matrix $\Gamma$~\cite{Briegel2001, Raussendorf2001}.  There are two equivalent ways to think of graph states. 
The first is operational in the sense that it provides a formula for preparing the state given a graph:
\begin{equation}\label{eq:graph_state}
    \ket{G} = \prod_{(a,b)\in E}U^{(a,b)}\ket{+}^{\otimes |V|},
\end{equation}
where 
\begin{equation}
    U^{(a,b)} = \dya{0}^{(a)}\otimes\mathbb{I}^{(b)}+\dya{1}^{(a)}\otimes\sigma_z^{(b)}
\end{equation}
is a controlled $Z$ operation on qubits $a$ and $b$, and 
\[
\ket{+}=\frac{1}{\sqrt{2}}(\ket{0}+\ket{1}), \quad
\ket{-}=\frac{1}{\sqrt{2}}(\ket{0}-\ket{1})
\]
forms the Hadamard basis.  Thus, given a graph, the graph state initialize $|V|$ qubits in the state $\ket{+}^{\otimes|V|}$ and, for each edge, apply a controlled Z between the corresponding qubits.

Graph states can be equivalently thought of as stabilizer states \cite{hein2004multiparty}. Here the stabilizers are $S_a = \sigma_x^{(a)}\prod_{b\in N_a}\sigma_z^{(b)}$, where $N_a$ is the neighborhood of vertex $a$. As there are $|V|$ qubits and stabilizers, $\ket{G}$ is the unique state stabilized by all $S_a$.

\begin{figure}[t]
\includegraphics[width=0.9\columnwidth]{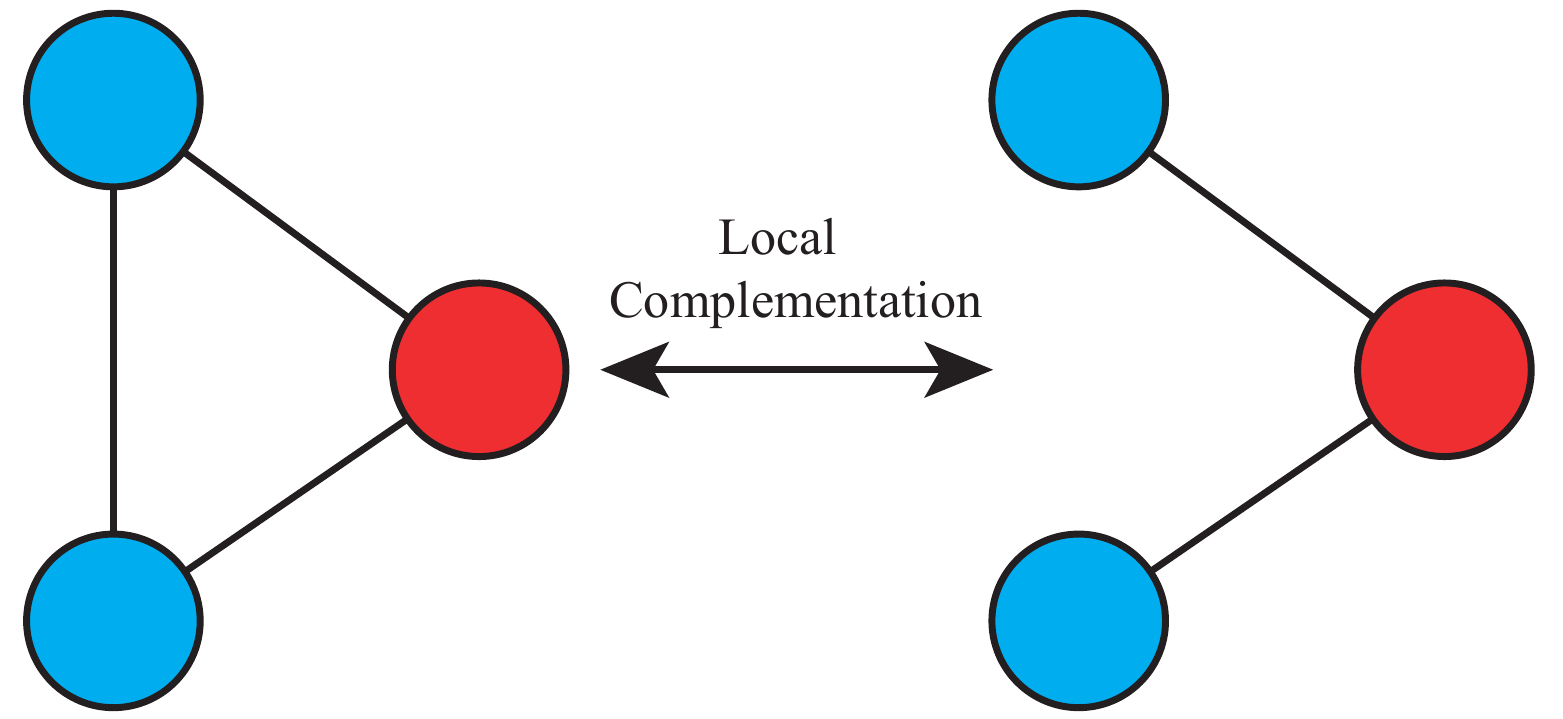}
\caption{Example of local complementation. Here the rule is applied to the red vertex, adding or removing the edge connecting the two other vertices.}
\label{fig:loc_comp}
\end{figure}

Also note that a basis for $\mathcal{H}=\bigotimes_{i=1}^n\mathcal{H}^{(i)}_2$ can be constructed given a graph $G$ \cite{Hein_app_ent_gs}:
\begin{equation}\label{eq:graph_basis}
    \ket{G_\textbf{s}} = \sigma_z^{\textbf{s}}\ket{G} = \prod_{(a,b)\in E}U^{(a,b)}\bigotimes_{i=1}^n(\sigma_z^{(i)})^{s_i}\ket{+}^{(i)}.
\end{equation}
It is clear that there are $2^n$ such orthogonal states and thus they form a basis. Further, one can think of $\textbf{s}$ as flipping the eigenvalues of stabilizers $S_a$ from $+1$ to $-1$. Going forward, we will denote these as graph basis states. We will later use a result from \cite{hein2004multiparty} that the partial trace of a graph state can be expressed in the graph basis:
\begin{align}\label{eq:g_ptrace}
    &\Tr_A[\ket{G}\bra{G}] = \nonumber\\ &\frac{1}{2^{|A|}}\sum_{\textbf{z}\in\mathbb{F}_2^{|A|}}U(\textbf{z})\ket{G-A}\bra{G-A}U(\textbf{z})^\dagger,
\end{align}
where $U(\textbf{z}) = \prod_{a\in A}(\prod_{b\in N_a} \sigma_z^{(b)})^{z_a}$, and $\ket{G-A}$ denotes the state corresponding to deleting all vertices in $A$ from $G$.  Further analysis of this state leads to the useful structural fact.
\begin{lemma}[\cite{hein2004multiparty}]
\label{thm:rdm}
    The states $U(\mbf{z})\ket{G-A}$ satisfy the orthogonality condition
    \begin{align}
        &\bra{G-A}U^\dagger(\mbf{z}')U(\mbf{z})\ket{G-A}=
        \nonumber\\
        &\begin{cases} 0\;\text{if $\mbf{z}-\mbf{z'}\in \text{\upshape ker}(\Gamma_{A\overline{A}})$} \\
        1\;\text{if $\mbf{z}-\mbf{z'}\not\in \text{\upshape ker}(\Gamma_{A\overline{A}})$}\end{cases},
    \end{align}
    where $\Gamma_{A\overline{A}}$ is the submatrix of $\Gamma_G$ restricted to edges from $A$ to $\overline{A}$.  Hence, $\rho^{(\overline{A})}=\Tr_A\op{G}{G}$ is maximally mixed over a subspace of dimension $2^d$, where $d=\text{\upshape rank}(\Gamma_{A\overline{A}})$. 
\end{lemma}

\subsection{Existing Tensor Rank Bounds for Graph States}
Here we briefly review existing results on graph state CP rank/Schmidt measure. From \cite{hein2004multiparty} we have that 
\begin{equation}\label{eq:gen_rk_lb}
    \rank(\ket{\psi}) \geq 2^{(\rank{\Gamma_{A\overline{A}})/2}},
\end{equation}
where $\Gamma_{A\overline{A}}$ is the subset of the adjacency matrix restricted to edges from $\overline{A}$ to $A$.  The authors also give a general case upper bound
\begin{equation}\label{eq:gen_rk_ub}
    \rank(\ket{\psi}) \leq 2^{\tau(G)},
\end{equation}
where $\tau(G)$ is the size of the smallest vertex cover of $G$.

\begin{figure}[t]
   \includegraphics[width=0.9\columnwidth]{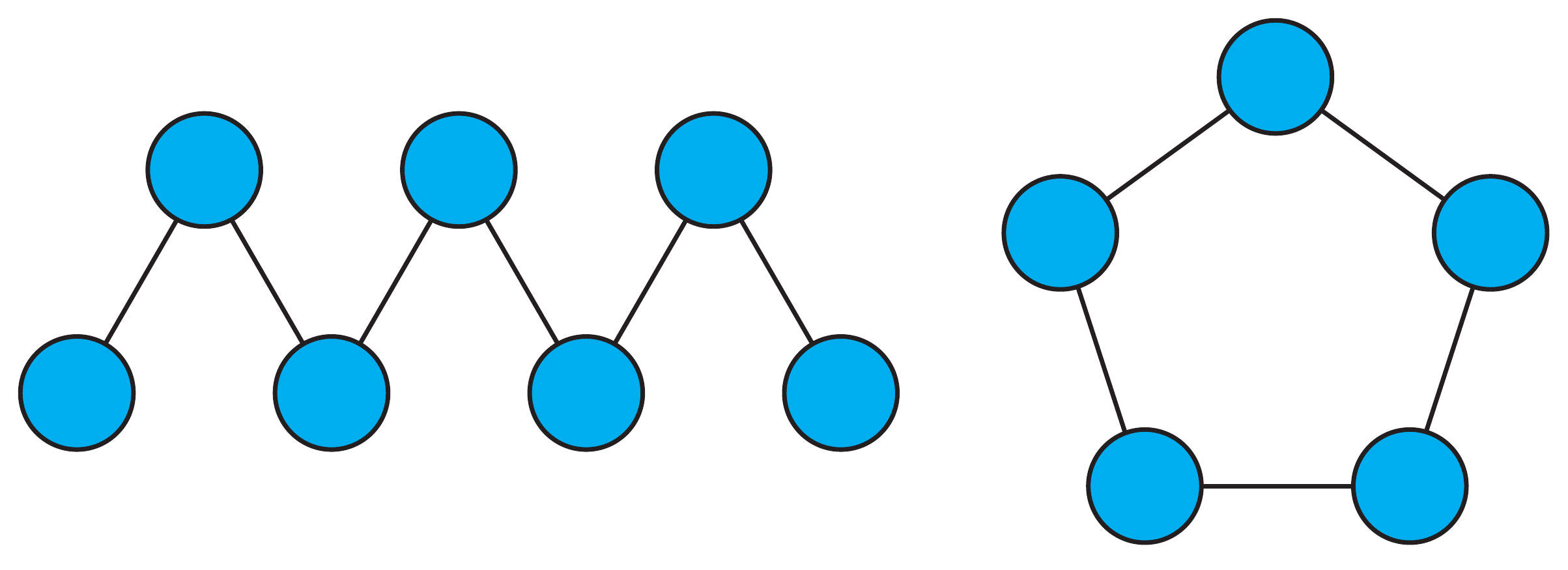}
    \caption{Line and ring graphs.
    The left graph is a line (one dimensional cluster state) on 7 qubits, which we denote by $\L{7}$. The right graph is an odd ring on 5 qubits, which we denote by $\R{5}$.}
    \label{fig:line_ring}
\end{figure}

While these bounds may not be tight, it is often possible to use complementation rules to find locally equivalent graphs for which these bounds improve.
It is known that the full orbit of any graph state under local clifford operations can be found via local complementations \cite{hein2004multiparty, Nest2004_loc_comp}. That is, for some vertex $a\in V$, complement the subgraph given by the neighborhood $N_a$ (Fig.~\ref{fig:loc_comp}). These rules have been used to classify all graph states of up to 8 qubits \cite{hein2004multiparty,Adcock2020_graph_ranks, Cabello2009_8q_schmidt}. Further, classes of two-colorable graphs corresponding to states of maximal schmidt measure are known \cite{Severini2006_2colorable}. However, odd rings, corresponding to non two-colorable graphs, lead to loose bounds.

Line states (Fig.~\ref{fig:line_ring}), also known as one-dimensional cluster states, are those with one-dimensional nearest neighbor connections. We will write $\L{n}$ to denote a line state on $n$ qubits. We will find line states to be useful in proving an upper bound on the rank of ring states. An explicit construction of a minimal CP decomposition of line states is given in the appendix.
\begin{lemma}
    \begin{align}
    \rank(\L{n}) = 2^{\lfloor \frac{n}{2} \rfloor}.
    \end{align}
\end{lemma}
\begin{proof}
    This readily follows from the mentioned graph theoretic tools. See \cite{hein2004multiparty} for details.
\end{proof}

For any even ring $\ket{R_{2n}}$, it's known that the lower bound equals the upper bound, thus the CP rank is $2^n$ \cite{hein2004multiparty}.
For any odd ring $\Rn$, it is known that $2^n \leq \rank (\Rn) \leq 2^{n+1}$, coming from the rank of the adjacency matrix and minimal vertex cover. Any tightening of these bounds will therefore require a new type of analysis not based on the latter graph-theoretic concepts.

\section{The Tensor Rank of Ring States}
In this paper we improve the odd ring CP rank bounds as stated in the following theorem.
\begin{theorem}\label{thm:odd_rank}
    The CP rank of the graph state corresponding to any odd ring $\Rn$, is bounded by
    \begin{equation}
    2^{n}+1\leq\rank(\Rn)\leq 3\cdot2^{n-1}.
    \end{equation}
\end{theorem}
\noindent We will break the proof into two propositions corresponding to the upper and lower bounds. To help clarify the arguments these are explicitly laid out for $\R{7}$ in the appendix.

\subsection{Upper Bound Analysis}\label{subsec:upper_bound}

In this section we provide a CP rank upper bound of $3\cdot2^{n-1}$ for odd ring graph states $\Rn$. Throughout the proof, we let 
\[
P_0 = \ket{0}\bra{0}, P_1 = \ket{1}\bra{1},
\]
and let $\CZ^{(a,b)}$ denote a controlled $Z$ operation where the $a$th qubit is the controlling qubit and the $b$th qubit is the controlled qubit. We have
\begin{align}\label{eq:CZ}
\CZ^{(a,b)} &= P_0^{(a)} \otimes I^{(b)} + P_1^{(a)} \otimes \sigma_z^{(b)} \nonumber\\
&= I^{(a)} \otimes \sigma_z^{(b)} + 2P_0^{(a)} \otimes P_1^{(b)}.
\end{align}
Below we show the main statement.

\begin{proposition}[CP rank upper bound for odd ring graph states]\label{prop:upper_bound}
    The CP rank of any odd ring state $\Rn$ is upper bounded by
    \begin{equation}\label{eq:upper_bound}
        \rank(\Rn) \leq 3\cdot2^{n-1}.
    \end{equation}
\end{proposition}
\begin{proof}
For the case with $n=1$, we can easily verify that \[\ket{R_{3}} = \ket{++-} + \frac{1}{\sqrt{2}}\ket{001}  - \frac{1}{\sqrt{2}}\ket{110},
\]
thus satisfying the upper bound. Below we show the cases with $n\geq 2$.

Based on \eqref{eq:CZ} and the fact that
\[\R{2n+1} = U^{(1,2n+1)}\L{2n+1},\]
we have 
\begin{align*}
    &\Rn = \CZ^{(1,2n+1)} \L{2n+1} \\
    &= I^{(1)} \otimes \sigma_z^{(2n+1)}\L{2n+1} + 2P_0^{(1)} \otimes P_1^{(2n+1)}\L{2n+1}.
\end{align*}
The CP rank of the term $I^{(1)} \otimes \sigma_z^{(2n+1)}\L{2n+1}$ is $2^n$ since the CP rank of $\L{2n+1}$ is $2^{n}$. 
Define the state
\begin{equation}\label{eq:phi}
    \ket{\phi_{2n+1}} =  P_0^{(1)}\L{2n+1}.
\end{equation}
Based on \cref{def:rank_ab_base} and \cref{def:rank_ab} below, the CP rank of the term $P_1^{(2n+1)}\ket{\phi_{2n+1}}$ for all integers $n\geq 2$ is upper-bounded by $2^{n-1}$, thus proving the statement.
\end{proof}

Below we present \cref{def:rank_ab_base} and \cref{def:rank_ab}, which upper-bound the CP rank of $P_0^{(2n+1)}\ket{\phi_{2n+1}}$ and $P_1^{(2n+1)}\ket{\phi_{2n+1}}$ for all integers $n\geq 2$.
In our analysis below, we define a generalized controlled gate \begin{align*}
 \CZZ^{(i,j,k)} &:= \CZ^{(i,j)}\CZ^{(i,k)}\\
 & = P_0^{(i)} \otimes I^{(j)} \otimes I^{(k)} + P_1^{(i)} \otimes \sigma_z^{(j)} \otimes \sigma_z^{(k)},
\end{align*}
whose CP rank is also 2.
The line state $\L{2n+1}$ can be expressed as
\begin{align}\label{eq:linestate}
  \L{2n+1} &= \prod_{i=1}^n \CZ^{(2i,2i-1)}\CZ^{(2i,2i+1)} \ket{+}^{(1,\ldots,2n+1)} \nonumber\\
&= \prod_{i=1}^n CZZ^{(2i,2i-1, 2i+1)} \ket{+}^{(1,\ldots,2n+1)}.
\end{align} 

\begin{lemma}\label{def:rank_ab_base}
When $n = 2$, the ranks of both $P_0^{(2n+1)}\ket{\phi_{2n+1}}$ and $P_1^{(2n+1)}\ket{\phi_{2n+1}}$  with $\ket{\phi_{2n+1}}$ defined in \eqref{eq:phi} are bounded by 2.
\end{lemma}
\begin{proof}
For $n=2$, 
{\small
\begin{align*}
&\L{2n +1} = \L{5} \\ &= CZZ^{(4,3,5)}CZZ^{(2,1,3)} 
\ket{+}^{(1,\ldots,5)} \\
  &  = \left(
I \otimes P_0 \otimes I \otimes P_0 \otimes I + 
I \otimes P_0 \otimes \sigma_z \otimes P_1 \otimes \sigma_z
\right)\ket{+}^{(1,\ldots,5)} \\
& +\left(
\sigma_z \otimes P_1 \otimes \sigma_z \otimes P_0 \otimes I + 
\sigma_z \otimes P_1 \otimes I \otimes P_1 \otimes \sigma_z
\right)\ket{+}^{(1,\ldots,5)} \\
& = \frac{1}{2} \ket{+0+0+} + 
\frac{1}{2} \ket{+0-1-} + 
\frac{1}{2} \ket{-1-0+} + 
\frac{1}{2} \ket{-1+1-},
\end{align*}}
thus we have
\begin{align*}
 & \ket{\phi_{5}} =  P_0^{(1)} \L{5} \\&= 
\frac{1}{2\sqrt{2}}\ket{0}\Big(
\ket{0+0+} + 
\ket{0-1-} + 
\ket{1-0+} + 
\ket{1+1-}
\Big) \\
& = \underbrace{\frac{1}{4}\ket{0}\Big(
\ket{0+0} + 
\ket{0-1} + 
\ket{1-0} + 
\ket{1+1}
\Big)\ket{0}}_{P_0^{(5)}\ket{\phi_{5}}}  \\ &
 + \underbrace{\frac{1}{4}\ket{0}\Big(
\ket{0+0} - 
\ket{0-1} +
\ket{1-0} -
\ket{1+1}
\Big)\ket{1}}_{P_1^{(5)}\ket{\phi_{5}}}.
\end{align*}
Above expressions for $P_0^{(5)}\ket{\phi_{5}}$ and $P_1^{(5)}\ket{\phi_{5}}$ can be rewritten as follows,
\begin{align*}
  &P_0^{(5)}\ket{\phi_{5}} = \frac{1}{2\sqrt{2}}\ket{0}\Big(
\ket{+0+} + \ket{-1-}
\Big)\ket{0}, \\
& P_1^{(5)}\ket{\phi_{5}} = \frac{1}{2\sqrt{2}}\ket{0}\Big(
\ket{+0-} + \ket{-1+}
\Big)\ket{1},
\end{align*}
thus the CP ranks are bounded by 2.
\end{proof}

\begin{lemma}\label{def:rank_ab}
When $n \geq 2$, the CP ranks of both states $P_0^{(2n+1)}\ket{\phi_{2n+1}}$ and $P_1^{(2n+1)}\ket{\phi_{2n+1}}$ are bounded by $2^{n-1}$.
\end{lemma}
\begin{proof}
We argue by induction on $n$. Assume that the ranks of both $P_0^{(2n+1)}\ket{\phi_{2n+1}}$ and $P_1^{(2n+1)}\ket{\phi_{2n+1}}$ are bounded by $2^{n-1}$. We will show that the CP ranks of both states $P_0^{(2n+3)}\ket{\phi_{2n+3}}$ and $P_1^{(2n+3)}\ket{\phi_{2n+3}}$ are bounded by $2^{n}$.

$\ket{\phi_{2n+3}}$ can be rewritten as follows,
\begin{align}\label{eq:induction}
  &\ket{\phi_{2n+3}} = P_0^{(1)} \L{2n+3} \nonumber\\
&= P_0^{(1)} CZZ^{(2n+2, 2n+1, 2n+3)}\L{2n+1}\ket{++}  \nonumber \\
& = CZZ^{(2n+2, 2n+1, 2n+3)}P_0^{(1)} \L{2n+1}\ket{++} \nonumber \\
& = CZZ^{(2n+2, 2n+1, 2n+3)} P_0^{(2n+1)}\ket{\phi_{2n+1}}\ket{++} \nonumber \\
&+ CZZ^{(2n+2, 2n+1, 2n+3)}P_1^{(2n+1)}\ket{\phi_{2n+1}}\ket{++} \nonumber\\
& = \frac{1}{\sqrt{2}}P_0^{(2n+1)}\ket{\phi_{2n+1}}\Big(\ket{0+} 
+ \ket{1-}\Big) \nonumber\\
& + 
\frac{1}{\sqrt{2}}P_1^{(2n+1)}\ket{\phi_{2n+1}}\Big(\ket{0+} - \ket{1-}\Big).
\end{align}
Note that the third equality comes from the commutativity of $CZZ^{(2n+2, 2n+1, 2n+3)}, P_0^{(1)}$. 
Based on the transformation
\begin{align*}
    \ket{0+} + \ket{1-} = \ket{+0} + \ket{-1}, \\ 
\ket{0+} - \ket{1-} = \ket{+1} + \ket{-0},
\end{align*}
\eqref{eq:induction} can be rewritten as
\begin{align*}
 &\ket{\phi_{2n+3}} \\ &= \frac{1}{\sqrt{2}}P_0^{(2n+1)}\ket{\phi_{2n+1}}\Big(\ket{+0} 
 + \ket{-1}\Big) \\
 &+ 
\frac{1}{\sqrt{2}}P_1^{(2n+1)}\ket{\phi_{2n+1}}\Big(\ket{+1} + \ket{-0}\Big)   \\
&= \underbrace{\frac{1}{\sqrt{2}} \Big(
P_0^{(2n+1)}\ket{\phi_{2n+1}}\ket{+} + P_1^{(2n+1)}\ket{\phi_{2n+1}}\ket{-}
\Big)
\ket{0}
}_{P_0^{(2n+3)}\ket{\phi_{2n+3}}}
\\
&+ 
\underbrace{\frac{1}{\sqrt{2}} \Big(
P_0^{(2n+1)}\ket{\phi_{2n+1}}\ket{-} + P_1^{(2n+1)}\ket{\phi_{2n+1}}\ket{+}
\Big)
\ket{1}
}_{P_1^{(2n+3)}\ket{\phi_{2n+3}}}
.
\end{align*}
It can be easily seen that the CP ranks of both states $P_0^{(2n+3)}\ket{\phi_{2n+3}}$ and $P_1^{(2n+3)}\ket{\phi_{2n+3}}$ are bounded by $2^{n}$.
Since the rank upper bounds for the base case ($n=2$) has been shown in \cref{def:rank_ab_base}, the lemma is proved. 
\end{proof}



\subsection{Lower Bound Analysis}

We now turn to the lower bound.  Recall that for any graph $G$ and subset of vertices $A\subset V$, the graph state $\ket{G}$ can be expressed via Eq.~\ref{eq:g_ptrace} as
\begin{align}
\label{Eq:graph-state-full}
    \ket{G}=\frac{1}{\sqrt{2^{|A|}}}\sum_{\mbf{z}\in\mbb{Z}_2^{|A|}}(-1)^{|\mbf{z}|}\ket{\mbf{z}}^{(A)}U(\mbf{z})\ket{G-A},
\end{align}
where $U(\textbf{z}) = \prod_{a\in A}(\prod_{b\in N_a} \sigma_z^{(b)})^{z_a}$. 
Here and later we represent $\mbf z$ as $ (z_2,z_4,\ldots,z_{2n})$.
Consider splitting the ring into $n$ even and $n+1$ odd vertices, denoting $A=\{2,4,\ldots,2n\}$ and $\overline{A}=\{1,3\ldots,2n+1\}$.
Now, the density matrix of $\overline{A}$ in $\ket{R_{2n+1}}$ is, using Eq.~\eqref{Eq:graph-state-full},
\begin{align*}
\rho^{(\overline{A})} = \frac{1}{2^{n}}\sum_{\mbf{z}\in\mbb{Z}_2^{n}}U(\mbf{z})\ket{R_{2n+1}-A}\bra{R_{2n+1}-A}U(\mbf{z})^{\dagger},
\end{align*}
since for each $a\in A$, $N_a\subseteq \overline{A}$.
The above density matrix decomposition is an eigendecomposition as a consequence of the following lemma.
\begin{lemma}
\label{Lem:basis-support}
The states 
\[ \ket{e_{\mbf{z}}}:=U(\mbf{z})\ket{R_{2n+1}-A}\qquad\forall \mbf{z}\in\mbb{Z}^{n}_2\]
are orthonormal and thus eigenvectors of $\rho^{(\overline{A})}$.
\end{lemma}
\begin{proof}


We can decompose the ring state into a product state and a two-qubit line state,
\[\ket{R_{2n+1}-A}=\bigotimes_{k=2,4,\ldots,2(n-1)}\ket{+}^{(k+1)}\ket{L_2^{(1,2n+1)}},\]
where $\ket{L_{2}^{(1,2n+1)}}=U^{(1,2n+1)}\ket{+}^{(1)}\otimes \ket{+}^{(2n+1)}$.
From the definition of $U(\mbf{z})$, for any $(z_2,z_4,\ldots,z_{2n})\in\mbb{Z}_2^n$ we can then write
\begin{align}
    \ket{e_{\mbf{z}}}=\ket{\phi_\mbf{z}}\otimes \left(\sigma_z^{z_2}\otimes\sigma_z^{z_{2n}}\ket{L_{2}^{(1,2n+1)}}\right)\notag,
    \end{align}
where
\begin{equation}
    \ket{\phi_\mbf{z}}=\bigotimes_{k=2,4,\ldots,2(n-1)}\sigma_z^{z_k\oplus z_{k+2}}\ket{+}^{(k+1)}.
\end{equation}
    We first show that the product states $\ket{\phi_\mbf z}$ are orthogonal except when $\mbf z' = \bar{\mbf z}$, where $\bar{\mbf z}$ denotes the bitwise conjugate of $\mbf{z}$.
Observe that, 
\begin{align*}
\langle \phi_{\mbf z}\ket{\phi_{\mbf z'}} =&
 \prod_{k=2,4,\ldots,2(n-1)} \bra{+}\sigma_z^{z_{k}\oplus z_{k+2}\oplus z_{k}' \oplus z_{k+2}'}\ket{+},
\end{align*}
hence  $\langle \phi_{\mbf z}\ket{\phi_{\mbf z'}} = 0$, unless $z_{k}\oplus z_{k}' = z_{k+2}\oplus z_{k+2}'$ for all $k=2,4,\ldots,2(n-1)$.
If each $z_{k}\oplus z_{k}'=0$, then $\mbf z' = \mbf z$, otherwise each $z_{k}\oplus z_{k}'=1$, so $\mbf z' = \bar{\mbf z}$.  Hence, we have established
\begin{equation}
\label{Eq:inner-product-z}
\ip{\phi_{\mbf{z}}}{\phi_{\mbf{z}'}}=\begin{cases}
0\qquad\text{if $\mbf{z}'\not=\overline{\mbf{z}}$}\\
1\qquad\text{if $\mbf{z}'=\overline{\mbf{z}}$}
\end{cases}.
\end{equation}
We now complete the proof of the lemma by showing that if $\mbf z'= \bar{\mbf z}$, the 2-vertex line-state components of $\ket{e_{\mbf z}}$ and $\ket{e_{\mbf z'}}$ are orthogonal,
\begin{align}
    &\bra{ L_2^{(1,2n+1)}}  \sigma_z^{z_2\oplus \overline{z}_2}\otimes\sigma_z^{z_{2n}\oplus\overline{z}_{2n}} \ket{L_2^{(1,2n+1)}} \nonumber \\
&= \frac{1}{2}(\bra{0+} + \bra{1-})(\sigma_z \otimes \sigma_z)(\ket{0+} + \ket{1-}) \nonumber \\
&= \frac{1}{2}(\bra{0+} + \bra{1-})(\ket{0-} - \ket{1+}) = 0.
\end{align}
\end{proof}

\begin{figure}
  \includegraphics[width=.5\linewidth]{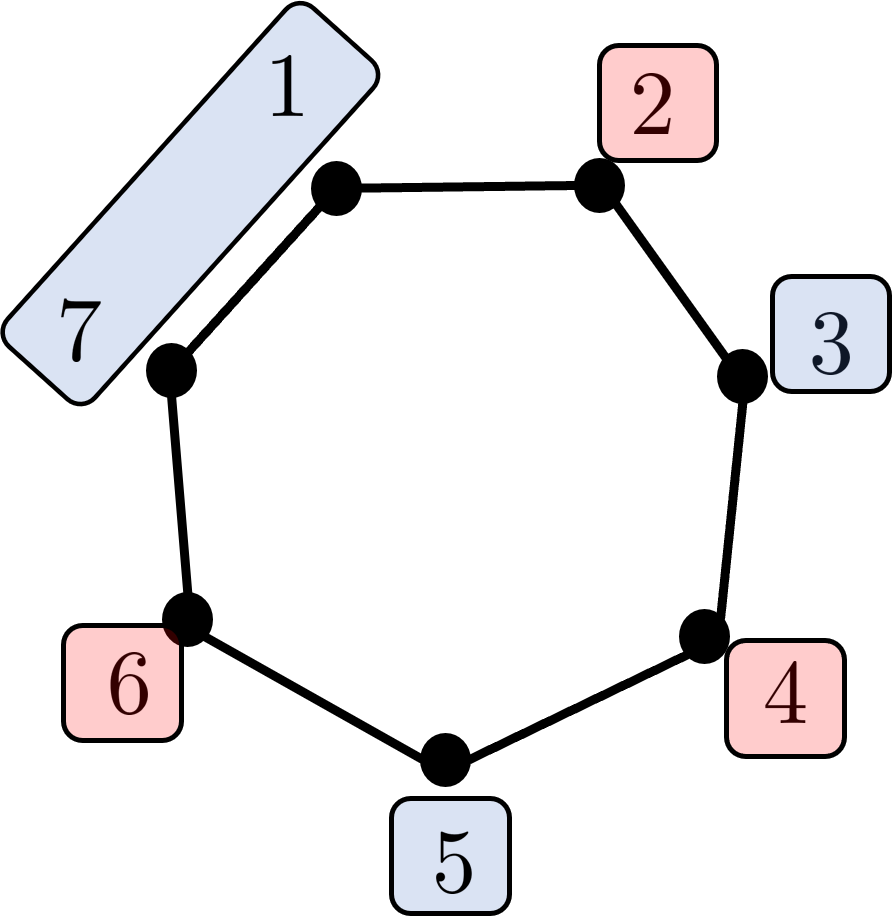}
  \caption{Lower bounding the rank of $\ket{R_{2n+1}}$ for $n=3$.  For the seven-qubit ring state, we trace out the subset of qubits $A=\{2,4,6\}$, depicted as the red nodes in the ring.  This set is judiciously chosen such that after tracing out qubits belonging to $A$, the remaining qubits $\overline{A}=\{1,3,5,7\}$ (shown in blue) consists of just one entangled pair $(1,7)$ and the rest are completely uncoupled.  This relatively simple structure allows us to characterize all the product states in the support of $\rho^{(\overline{A})}$ (Lemma \ref{Lem:product-states-support}).  Using this characterization, the desired lower bound on the tensor rank is proven (Theorem \ref{Thm:rank-lb}).}
  \label{fig:ring-pic}
\end{figure}

We next go further and characterize the product states in the support of $\rho^{(\overline{A})}$.
\begin{lemma}
\label{Lem:product-states-support}
Let $S_{0,0}\subset\mbb{Z}_2^n$ be the collection of sequences $\mbf{z}=(z_2,z_4,\ldots,z_{2n})$ with $z_2=0$ and $z_{2n}=0$, and let $S_{0,1}$ be the collection of sequences $\mbf{z}=(z_2,z_4,\ldots,z_{2n})$ with $z_2=0$ and $z_{2n}=1$.  Then the support of $\rho^{(\overline{A})}$ contains only $2^n$ product states given by
\begin{align}
\label{Eq:product-state-support}
    \forall\mbf{z}&\in S_{0,0}:\notag\\
    &\qquad\begin{cases}\frac{1}{\sqrt{2}}\left(e^{i\pi/4}\ket{e_{\mbf{z}}}+e^{-i\pi/4}\ket{e_{\mbf{\overline{z}}}}\right)=\ket{\tilde{+}\tilde{+}}\ket{\phi_{\mbf{z}}}\\
    \frac{1}{\sqrt{2}}\left(e^{-i\pi/4}\ket{e_{\mbf{z}}}+e^{i\pi/4}\ket{e_{\overline{\mbf{z}}}}\right)=\ket{\tilde{-}\tilde{-}}\ket{\phi_{\mbf{z}}}\end{cases}\notag\\
    \forall\mbf{z}&\in S_{0,1}:\notag\\
    &\qquad\begin{cases}\frac{1}{\sqrt{2}}\left(e^{i\pi/4}\ket{e_{\mbf{z}}}+e^{-i\pi/4}\ket{e_{\mbf{\overline{z}}}}\right)=\ket{\tilde{+}\tilde{-}}\ket{\phi_{\mbf{z}}}\\
    \frac{1}{\sqrt{2}}\left(e^{-i\pi/4}\ket{e_{\mbf{z}}}+e^{i\pi/4}\ket{e_{\overline{\mbf{z}}}}\right)=\ket{\tilde{-}\tilde{+}}\ket{\phi_{\mbf{z}}}\end{cases}
\end{align}
where $\ket{\tilde{\pm}}=\frac{1}{\sqrt{2}}(\ket{0}\pm i\ket{1})$.   
\end{lemma}
\begin{proof}
By Lemma \ref{Lem:basis-support}, the support of $\rho^{(\overline{A})}$ is spanned by the orthonormal states $\ket{e_{\mbf{z}}}$.  Then since $\ket{\phi_{\mbf{z}}}=\ket{\phi_{\overline{\mbf{z}}}}$, we can organize the $2^n$ eigenstates $\ket{e_{\mbf{z}}}$ into four sets as
\begin{align}
\label{Eq:eigenstates-grouped}
    &z_2=0, z_{2n}=0:\notag\\
    &\left\{\tfrac{1}{\sqrt{2}}\left(\ket{0+}+\ket{1-}\right)^{(1,2n+1)}\ket{\phi_{\mbf{z}}}^{(3,5,\ldots, 2n-1)} \mbf{z}\in S_{0,0}\right\}\notag\\
    &z_2=1, z_{2n}=1:\notag\\
    &\left\{\tfrac{1}{\sqrt{2}}\left(\ket{0-}-\ket{1+}\right)^{(1,2n+1)}\ket{\phi_{\mbf{z}}}^{(3,5,\ldots, 2n-1)} \mbf{z}\in S_{0,0}\right\}\notag\\
    &z_2=0, z_{2n}=1:\notag\\
    &\left\{\tfrac{1}{\sqrt{2}}\left(\ket{0-}+\ket{1+}\right)^{(1,2n+1)}\ket{\phi_{\mbf{z}}}^{(3,5,\ldots, 2n-1)} \mbf{z}\in S_{0,1}\right\}\notag\\
    &z_2=1, z_{2n}=0:\notag\\
    &\left\{\tfrac{1}{\sqrt{2}}\left(\ket{0+}-\ket{1-}\right)^{(1,2n+1)}\ket{\phi_{\mbf{z}}}^{(3,5,\ldots, 2n-1)} \mbf{z}\in S_{0,1}\right\}.
\end{align}
The crucial observation is that the states in the first two sets can be written using only two product states for qubits $2$ and $3$, and similarly for the states in the second two sets:
\begin{align}
    \frac{1}{\sqrt{2}}\left(\ket{0+}+\ket{1-}\right)&=\frac{1}{\sqrt{2}}\left(e^{-i\pi/4}\ket{\tilde{+}\tilde{+}}+e^{i\pi/4}\ket{\tilde{-}\tilde{-}}\right)\notag\\
    \frac{1}{\sqrt{2}}\left(\ket{0-}-\ket{1+}\right)&=\frac{1}{\sqrt{2}}\left(e^{i\pi/4}\ket{\tilde{+}\tilde{+}}+e^{-i\pi/4}\ket{\tilde{-}\tilde{-}}\right)\notag\\
    \frac{1}{\sqrt{2}}\left(\ket{0-}+\ket{1+}\right)&=\frac{1}{\sqrt{2}}\left(e^{-i\pi/4}\ket{\tilde{+}\tilde{-}}+e^{i\pi/4}\ket{\tilde{-}\tilde{+}}\right)\notag\\
    \frac{1}{\sqrt{2}}\left(\ket{0-}-\ket{1+}\right)&=\frac{1}{\sqrt{2}}\left(e^{i\pi/4}\ket{\tilde{+}\tilde{-}}+e^{-i\pi/4}\ket{\tilde{-}\tilde{+}}\right).
\end{align}
Therefore, we conclude that the support of $\rho^{\overline{A}}$ is spanned by $2^n$ orthogonal product states
\begin{align}
    &\left\{\ket{\tilde{+}\tilde{+}}\ket{\phi_{\mbf{z}}},\;\ket{\tilde{-}\tilde{-}}\ket{\phi_{\mbf{z}}}\;\vert\;\mbf{z}\in S_{0,0}\right\}\notag\\
    \bigcup\; &\left\{\ket{\tilde{+}\tilde{-}}\ket{\phi_{\mbf{z}}},\;\ket{\tilde{-}\tilde{+}}\ket{\phi_{\mbf{z}}}\;\vert\;\mbf{z}\in S_{0,1}\right\}.
\end{align}
We next show that these are the only product states in the support of 
$\rho^{(\overline A)}$.

Suppose that $\ket{\Psi}$ is a product state in the support of $\rho^{(\overline A)}$.  Then we can find coefficients such that
\begin{align}
\label{Eq:Psi-product-expansion}
    \ket{\Psi}&=\sum_{\mbf{z}\in S_{0,0}}\left(a_{\mbf{z}}\ket{\tilde{+}\tilde{+}}\ket{\phi_{\mbf{z}}}+b_{\mbf{z}}\ket{\tilde{-}\tilde{-}}\ket{\phi_{\mbf{z}}}\right)\notag\\
    &+\sum_{\mbf{z}\in S_{0,1}}\left(c_{\mbf{z}}\ket{\tilde{+}\tilde{-}}\ket{\phi_{\mbf{z}}}+d_{\mbf{z}}\ket{\tilde{-}\tilde{+}}\ket{\phi_{\mbf{z}}}\right)\notag\\
    &=\ket{\tilde{+}\tilde{+}}\ket{\alpha}+\ket{\tilde{-}\tilde{-}}\ket{\beta}+\ket{\tilde{+}\tilde{-}}\ket{\gamma}+\ket{\tilde{-}\tilde{+}}\ket{\delta},
\end{align}
where $\ket{\alpha}=\sum_{\mbf{z}\in S_{0,0}}a_{\mbf{z}}\ket{\phi_{\mbf{z}}}$, $\ket{\beta}=\sum_{\mbf{z}\in S_{0,0}}b_{\mbf{z}}\ket{\phi_{\mbf{z}}}$, $\ket{\gamma}=\sum_{\mbf{z}\in S_{0,1}}c_{\mbf{z}}\ket{\phi_{\mbf{z}}}$, and $\ket{\delta}=\sum_{\mbf{z}\in S_{0,1}}d_{\mbf{z}}\ket{\phi_{\mbf{z}}}$.  By Eq. \eqref{Eq:inner-product-z}, the states $\{\ket{\alpha},\ket{\beta}\}$ are orthogonal to the states $\{\ket{\gamma},\ket{\delta}\}$.  

Suppose first that both $\ket{\alpha}$ and $\ket{\beta}$ are nonzero.  Then we can find a vector $\ket{v}$ in the linear span of $\{\ket{\alpha},\ket{\beta}\}$ that has nonzero overlap with both $\ket{\alpha}$ and $\ket{\beta}$.  Partially contracting both sides of Eq. \eqref{Eq:Psi-product-expansion} by $\ket{v}$ yields
\[(\bra{v})\ket{\Psi}=x\ket{\tilde{+}\tilde{+}}+y\ket{\tilde{-}\tilde{-}}\]
with $x,y\not=0$.  But since $\ket{\Psi}$ is a product state, it remains a product state under partial contraction, and so the RHS must be a product state.  However the only product states contained in the linear span of $\ket{\tilde{+}\tilde{+}}$ and $\ket{\tilde{-}\tilde{-}}$ are these themselves.  We thus have a contradiction, and so it is not possible for both $\ket{\alpha}$ and $\ket{\beta}$ to be nonzero.  A similar argument shows that both $\ket{\gamma}$ and $\ket{\delta} $ cannot be nonzero.  Hence, there are only four possible forms of $\ket{\Psi}$, each pairing an element in $\{\ket{\alpha},\ket{\beta}\}$ with an element in $\{\ket{\gamma},\ket{\delta}\}$.  For example, we have could have
\begin{align*}
    \ket{\Psi}&=\ket{\tilde{+}\tilde{+}}\ket{\alpha}+\ket{\tilde{+}\tilde{-}}\ket{\delta}=\ket{\tilde{+}}(\ket{\tilde{+}}\ket{\alpha}+\ket{\tilde{-}}\ket{\delta}),
\end{align*}
which is not a product state unless either $\ket{\alpha}$ or $\ket{\delta}$ are zero, since $\ip{\alpha}{\delta}=0$. A similar argument applies for the other three possible forms of $\ket{\Psi}$.  Therefore, any product state in $\ket{\Psi}$ must have the form $\ket{\tilde{+}\tilde{+}}\ket{\alpha}$, $\ket{\tilde{-}\tilde{-}}\ket{\beta}$, $\ket{\tilde{+}\tilde{-}}\ket{\gamma}$, or $\ket{\tilde{-}\tilde{-}}\ket{\delta}$, where $\ket{\alpha},\ket{\beta}$ are product states in the span of $\{\ket{\phi_{\mbf{z}}}\;|\;\mbf{z}\in S_{0,0}\}$ and $\ket{\gamma},\ket{\delta}$ are product states in the span of $\{\ket{\phi_{\mbf{z}}}\;|\;\mbf{z}\in S_{0,1}\}$.

Then, it finally remains to be shown that the only product states in the span of $\{\ket{\phi_{\mbf{z}}}\;|\;\mbf{z}\in S_{0,0}\}$ are the $\ket{\phi_{\mbf{z}}}$ themselves;  likewise, the only product states in the span of $\{\ket{\phi_{\mbf{z}}}\;|\;\mbf{z}\in S_{0,1}\}$ are the $\ket{\phi_{\mbf{z}}}$ themselves.

Suppose that $\ket{\varphi}=\sum_{\mbf{z}\in S_{0,0}}a_{\mbf{z}}\ket{\phi_{\mbf{z}}}$ is a product state.  If $n=2$, then there is only a single state $\ket{+}^{(3)}$ in this sum.  If $n=3$, then there are two terms,
\[\ket{\varphi}=a_0\ket{++}^{(3,5)}+a_1\ket{--}^{(3,5)},\]
which requires that either $a_0=0$ or $a_1=0$ in order for $\ket{\varphi}$ to be a product state.  Next we consider the case when $n>3$.  For any binary sequence $\mbf{w}=(w_8,w_{10},\ldots,w_{2n})$, define the $(n-3)$-qubit state
\begin{align}
    &\ket{\omega_{\mbf{w}}}^{(7,\ldots, 2n-1)}\nonumber \\&:=\bigotimes_{k=6,8,\ldots,2(n-1)}\left(\sigma_z^{(k+1)}\right)^{w_k\oplus z_{w+2}}\ket{+}^{(k+1)},
\end{align}
where $w_{2(n-1)}=0$.  By Eq. \eqref{Eq:inner-product-z}, it follows that for any $\ket{\phi_{\mbf{z}}}^{(3,5,\ldots, 2n-1)}$ with $\mbf{z}\in S_{0,0}$ we have the partial contractions
\begin{align}
    \ip{\omega_\mbf{w}}{\phi_{\mbf{z}}}=\begin{cases} \text{$\ket{++}^{(3,5)}$ or $\ket{--}^{(3,5)}$ if $w_6=0$}\\\text{ $\ket{+-}^{(3,5)}$ or $\ket{-+}^{(3,5)}$ if $w_6=1$}\end{cases}.
\end{align}
Therefore, $\ip{\omega_\mbf{w}}{\varphi}$ is either contained in the linear span of $\{\ket{++},\ket{--}\}$ or $\{\ket{+-},\ket{-+}\}$.  In both cases, there are no other product states in the respective spaces besides the given ones.  Thus, if $\ket{\varphi}=\sum_{\mbf{z}\in S_{0,0}}a_{\mbf{z}}\ket{\phi_{\mbf{z}}}$ is a product state, then it requires that one and only one of the $a_\mbf{z}$ be nonzero.  An analogous argument holds for the superposition states $\ket{\varphi}=\sum_{\mbf{z}\in S_{0,1}}a_{\mbf{z}}\ket{\phi_{\mbf{z}}}$.  This concludes the proof.
\end{proof}

Lemma \ref{Lem:product-states-support} provides a structural analysis of the ring state $\ket{R_{2n+1}}$ that we will use to lower bound the tensor rank of $\ket{R_{2n+1}}$.  To get to this we will also need one general fact about CP decompositions.

\begin{lemma}
\label{Lem:cpd-rank}
Suppose that
\begin{equation}
\label{Eq:cpd2}
    \ket{\psi} = \sum_{i=1}^R \mu_i\ket{\psi_i}^{(1)}\otimes \cdots \otimes \ket{\psi_i}^{(N)},
\end{equation}
is a CP decomposition of $\ket{\psi}$.  For any subset of parties $\overline{A}$, the states $\{\bigotimes_{c\in \overline{A}}\ket{\psi_i}^{(c)}\}_{i=1}^R$ contains the support of $\rho^{(\overline{A})}=\Tr_A\op{\psi}{\psi}$. 
Moreover, if $\rho^{(\overline{A})}$ has matrix rank $R$, then conversely the states $\{\bigotimes_{c\in \overline{A}}\ket{\psi_i}^{(c)}\}_{i=1}^R$ must belong to the support of $\rho^{(\overline{A})}$.
\end{lemma}
\begin{proof}

Given Eq.~\eqref{Eq:cpd2}, the reduced density matrix of $\ket \psi$ on $\bar{A}$ is
\[\rho^{(\overline{A})}=
\sum_{i=1}^R \sum_{j=1}^R\underbrace{\mu_i\mu_j\prod_{k\in A}\langle \psi_i^{(k)}\ket{\psi_j^{(k)}}}_{M_{ij}} \prod_{c\in \overline{A}} \ket{\psi^{(c)}_i}\bra{\psi^{(c)}_j}
.\]
Hence, $\rho^{(\overline{A})}= UMU^T$, where
\[U =\begin{bmatrix}\bigotimes_{c\in \overline{A}}\ket{\psi^{(c)}_1} & \cdots &\bigotimes_{c\in \overline{A}}\ket{\psi^{(c)}_R}\end{bmatrix}.\]
Consequently, the support (column span) of this reduced density matrix is contained in $\spn(\{\bigotimes_{c\in \overline{A}}\ket{\psi^{(c)}_i}\}_{i=1}^R)$.
Further, if $\rank(\rho^{(\overline A)})=R$, the rank of $M$ is also $R$, and the column span of $\rho^{(\overline A)}$ is the same if as that of $U$.


\end{proof}

Now, we put everything together to obtain our desired lower bound.
\begin{theorem}
\label{Thm:rank-lb}
$\text{rank}(\ket{R_{2n+1}})>2^n$.
\end{theorem}
\begin{proof}
Lemma \ref{Lem:cpd-rank} constructs a reduced density matrix from $\ket{R_{2n+1}}$ of rank $2^n$, hence $\text{rank}(\ket{R_{2n+1}})\geq 2^n$.
Now, suppose for sake of contradiction that $\text{rank}(\ket{R_{2n+1}})=2^n$.
Since $\rho^{(\overline{A})}$ has matrix rank $2^n$ for the subset $A=\{2,4,\ldots,2n\}$, Lemma \ref{Lem:cpd-rank} says that any CP decomposition of $\ket{R_{2n+1}}$ of minimal length must contain product states $\{\bigotimes_{c\in \overline{A}}\ket{\psi_i}^{(c)}\}_{i=1}^R$ belonging to the support of $\rho^{(\overline{A})}$.  However, Lemma \ref{Lem:product-states-support} then implies that these product states must be the ones given in \eqref{Eq:product-state-support}.  That is, we must be able to write
{\small
\begin{align}
    \ket{R_{2n+1}}&=\sum_{\mbf{z}\in S_{0,0}}\ket{A_{\mbf{z}}}^{(2,4,\ldots, 2n)}\ket{\tilde{+}\tilde{+}}^{(1,2n+1)}\ket{\phi_{\mbf{z}}}^{(3,5,\ldots, 2n-1)}\notag\\
    &+\sum_{\mbf{z}\in S_{0,0}}\ket{B_{\mbf{z}}}^{(2,4,\ldots, 2n)}\ket{\tilde{-}\tilde{-}}^{(1,2n+1)}\ket{\phi_{\mbf{z}}}^{(3,5,\ldots, 2n-1)}\notag\\
    &+\sum_{\mbf{z}\in S_{0,1}}\ket{C_{\mbf{z}}}^{(2,4,\ldots, 2n)}\ket{\tilde{+}\tilde{-}}^{(1,2n+1)}\ket{\phi_{\mbf{z}}}^{(3,5,\ldots, 2n-1)}\notag\\
    &+\sum_{\mbf{z}\in S_{0,1}}\ket{D_{\mbf{z}}}^{(2,4,\ldots, 2n)}\ket{\tilde{-}\tilde{+}}^{(1,2n+1)}\ket{\phi_{\mbf{z}}}^{(3,5,\ldots, 2n-1)}
    \label{Eq:ring-expansion}
\end{align}}\noindent
with $\ket{A_{\mbf{z}}}$, $\ket{B_{\mbf{z}}}$, $\ket{C_{\mbf{z}}}$, and $\ket{D_{\mbf{z}}}$ all being product states.

At the same time, from Lemma \ref{Lem:basis-support} and Eq. \eqref{Eq:graph-state-full}, we can express the ring state as
\begin{equation}
    \ket{R_{2n+1}}=\frac{1}{\sqrt{2^{n}}}\sum_{\mbf{z}\in\mbb{Z}_2^{n}}\ket{\hat{\mbf{z}}}^{(A)}\ket{e_{\mbf{z}}}^{(\overline{A})},
\end{equation}
where $\ket{\hat{\mbf{z}}}^{(A)}:=(-1)^{|\mbf{z}|}\ket{\mbf{z}}^{(A)}$.  By inverting the equalities in Eq. \eqref{Eq:product-state-support}, this can be written as
\begin{align}
    \ket{R_{2n+1}}
    &=\frac{1}{\sqrt{2^{n}}}\sum_{\mbf{z}\in S_{0,0}}\frac{e^{-i\pi/4}\ket{\hat{\mbf{z}}}+e^{i\pi/4}\ket{\hat{\overline{\mbf{z}}}}}{\sqrt{2}}\ket{\tilde{+}\tilde{+}}\ket{\phi_{\mbf{z}}}\notag\\
    &+\frac{1}{\sqrt{2^{n}}}\sum_{\mbf{z}\in S_{0,0}}\frac{e^{i\pi/4}\ket{\hat{\mbf{z}}}+e^{-i\pi/4}\ket{\hat{\overline{\mbf{z}}}}}{\sqrt{2}}\ket{\tilde{-}\tilde{-}}\ket{\phi_{\mbf{z}}}\notag\\
    &+\frac{1}{\sqrt{2^{n}}}\sum_{\mbf{z}\in S_{0,1}}\frac{e^{-i\pi/4}\ket{\hat{\mbf{z}}}+e^{i\pi/4}\ket{\hat{\overline{\mbf{z}}}}}{\sqrt{2}}\ket{\tilde{+}\tilde{-}}\ket{\phi_{\mbf{z}}}\notag\\
    &+\frac{1}{\sqrt{2^{n}}}\sum_{\mbf{z}\in S_{0,1}}\frac{e^{i\pi/4}\ket{\hat{\mbf{z}}}+e^{-i\pi/4}\ket{\hat{\overline{\mbf{z}}}}}{\sqrt{2}}\ket{\tilde{-}\tilde{+}}\ket{\phi_{\mbf{z}}}.\notag
\end{align}
Comparing this with Eq. \eqref{Eq:ring-expansion} shows that $\ket{A_{\mbf{z}}}=\frac{1}{\sqrt{2^{n+1}}}\left(e^{-i\pi/4}\ket{\hat{\mbf{z}}}+e^{i\pi/4}\ket{\hat{\overline{\mbf{z}}}}\right)$, $\ket{B_{\mbf{z}}}=\frac{1}{\sqrt{2^{n+1}}}\left(e^{i\pi/4}\ket{\hat{\mbf{z}}}+e^{-i\pi/4}\ket{\hat{\overline{\mbf{z}}}}\right)$, etc., which is a contradiction because these are not product states.

\end{proof}

\section{Entanglement Measures on General Graph States}

\subsection{Extensions of Bipartite Measures}
Here we evaluate the multipartite extensions of bipartite measures previously discussed.  Recall that the multipartite extension of some bipartite entanglement measure $E$ is defined as $\min_{A}E_A(\ket{\psi})$, where $E_A$ is the measure $E$ evaluated according to partition $A|\overline{A}$.  Surprisingly for graph states, the multipartite extensions of many standard bipartite entanglement measures are dichotomous: one of two values based on if the graph is connected. The GME concurrence, negativity, and geometric measure have been previously calculated for connected graphs \cite{Contreras-Tejada-2019a}.  Since these results were derived independently of this work, we provide in the appendix a self-contained and direct calculation of the following:
\begin{theorem}
\label{thm:GME}
    \begin{align}
         \cgme{\ket{G}}&=\begin{cases}
    0 & \text{if G is a disconnected graph} \\
    1 & \text{otherwise}
    \end{cases}\\
    \ngme{\ket{G}}&=\begin{cases}
    0 & \text{if G is a disconnected graph} \\
    \frac{1}{2} & \text{otherwise}
    \end{cases}\\
    \ggme{\ket{G}} &= \begin{cases}
    0 & \text{if G is a disconnected graph} \\
    \frac{1}{2} & \text{otherwise}
\end{cases}.
\end{align}
\end{theorem}

\begin{figure}[t]
`   \includegraphics[width=0.9\columnwidth]{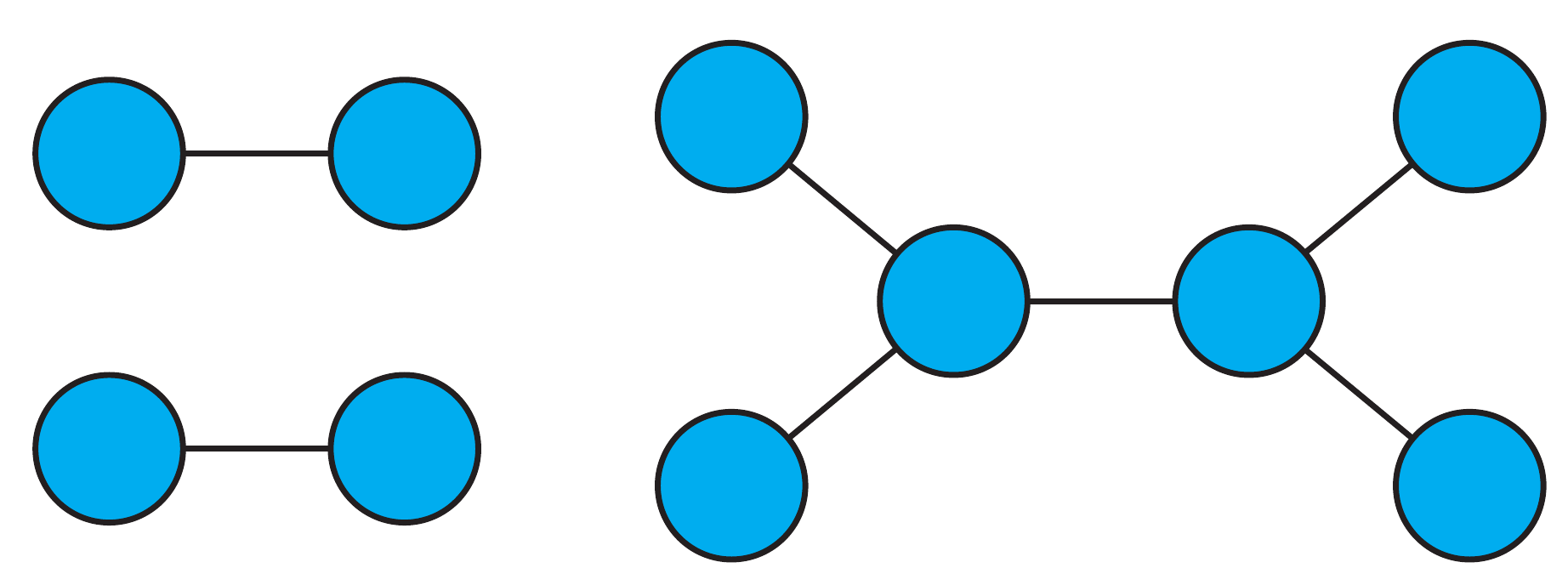}
    \caption{Graphs with $\tau_n=1$. All vertices in the graphs above have odd degree and thus have maximal n-tangle via Thm.~\ref{thm:n_tangle}.}
    \label{fig:my_label}
\end{figure}

\subsection{Graph State N-Tangle}

We present the following dichotomous result for $\tau_n(\ket{G})$:
\begin{theorem}\label{thm:n_tangle}
    For any graph state $\ket{G}$ on an even number of qubits,
    \begin{align}
        \tau_n(\ket{G}) = \begin{cases}
            1 & \forall v\in V \ \delta(v) = 1 \text{ (mod 2)}\\
            0 & \text{otherwise}
        \end{cases},
    \end{align}
    where $\delta(v)$ denotes the degree of vertex v.
\end{theorem}
\begin{proof}
    As the components of the state vector $\ket{G}$ are all real, $\tau_n(\ket{G}) = |\braket{G|\sigma_y^{\otimes n}|G}|^2$. We employ the following $\sigma_x$ rule \cite{Griffiths}:
    \begin{align}
        \sigma_x^{(a)} \ket{G} = \prod_{b\in N_a}\sigma_z^{(b)}\ket{G}.
    \end{align}
    Thus, applying $\sigma_x^{(a)}$ maps the state to another graph basis state based on the graph's edges. Of course, $\sigma_y = -i\sigma_z\sigma_x$. Note that additional global phases may be picked up in the application of the $\sigma_x$ rule based on the graph basis state. However, global phases do not factor into the calculation of $\tau_n$ and are dropped below. Thus,
    \begin{align}
        \sigma_y^{\otimes n}\ket{G} \sim \prod_{a\in V}\sigma_z^{(a)}\prod_{b\in N_a}\sigma_z^{(b)}\ket{G}.
    \end{align}
    If there is a vertex $v$ of even degree, then $\sigma_z^{(v)}$ appears in an even number of the second products above in addition to once in the first product. Thus, $\ket{G}$ is mapped to a different graph basis state and $\lvert\braket{G|\sigma_y^{\otimes n}|G}\rvert=0$.
    
    If all vertices have odd degree, then an even number of $\sigma_z$ operations are applied to each party, and the state is an eigenvector of $\sigma_y^{\otimes n}$. Thus, $\lvert\braket{G|\sigma_y^{\otimes n}|G}\rvert=1$.
\end{proof}

\begin{figure}[t]
    \centering
    \includegraphics[width=0.5\columnwidth]{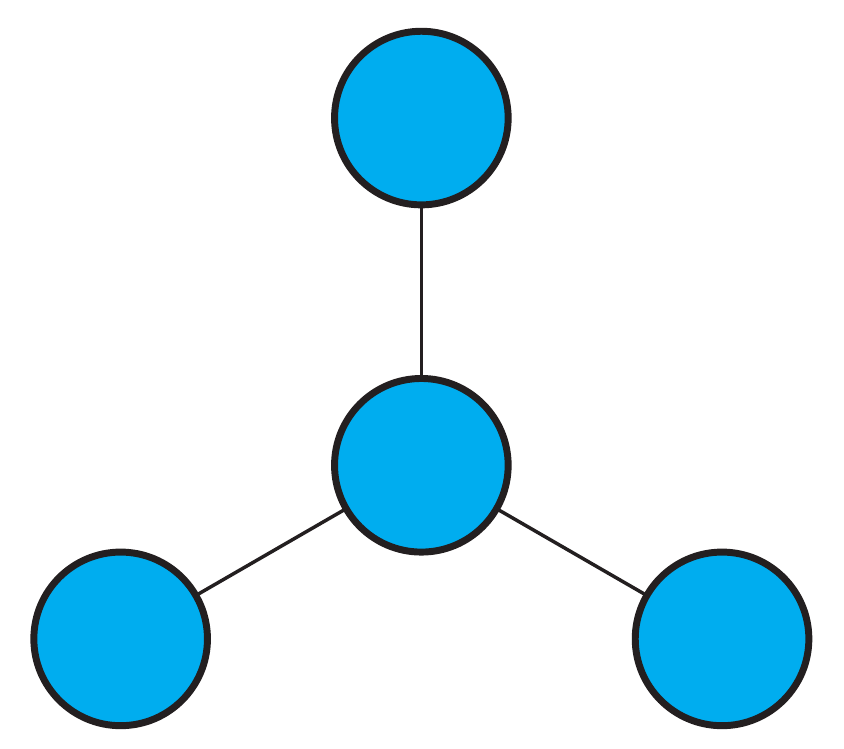}
    \caption{GHZ state graph. For any system size, the GHZ state is local Clifford equivalent to a star graph state. Through the local complementation rule, this is also local Clifford equivalent to the fully connected graph.}
    \label{fig:ghz}
\end{figure}

\begin{remark}
    There are $2^{\binom{n-1}{2}}$ graphs on $n$ (where $n$ is even) vertices such that $\tau_n(\ket{G}) = 1$. There are none for odd $n$.
\end{remark}
\begin{proof}
    We claim that for even $n$, given any graph on $n-1$ vertices we can construct a graph on $n$ of all odd degree. First recall the result from elementary graph theory that the number of vertices of odd degree must be even. As $n-1$ is odd, thus there are an odd number of vertices of even degree. Simply connect the new $n^{\text{th}}$ vertex to these originally even degree vertices. As the original graph can be recovered by removing the $n^{\text{th}}$ vertex, this is a bijection between graphs of $n-1$ vertices and all odd graphs on $n$. The claim follows from the number of graphs on $n-1$ labeled vertices.
    
    As there cannot be an odd number of vertices of odd degree, $\tau_n(\ket{G}) = 0$ for all graphs on an odd number of vertices. This is exactly what we would like to recover as it is known that n-tangle is only nonzero for even numbers of qubits \cite{Wong2001_ntangle}.
    

\end{proof}

We also note that this recovers the fact that $\tau_n(\ket{\text{GHZ}})=1$. A fully connected or star graph (Fig.~\ref{fig:ghz}) corresponds to a GHZ state (up to local unitary operations). As expected, all vertices have odd degree.

\begin{remark}
    The n-tangle $\tau_n$ of any stabilizer state can be efficiently computed via stabilizer computation.
\end{remark}
This is clear from stabilizer states taking the form $\ket{s}\bra{s} = \sum_i g_i$, where $\{g_i\}$ are the n stabilizers. The complex conjugate of the state can readily be found via transposing the stabilizers with respect to the computational basis, $\ket{s^*}\bra{s^*} = \sum_i T(g_i)$, which simply adds a sign of $-1$ to $\sigma_y$ terms and leaves all other stabilizer unchanged. Thus, one can continue by calculating $g' = \{g_i'| g_i'=\sigma_y^{\otimes n}g_i\}$ and checking if $\langle g\rangle =\langle g'\rangle$.


\section{Overview and Discussion}
In this work we tightened the bounds on CP rank of odd rings to $2^n+1 \leq \rank(\R{2n+1}) \leq 3\cdot 2^{n-1}$. This indicates that odd rings are, according to the Schmidt measure, more entangled than a line in the same number of qubits. Further, odd rings are thus not of particularly high rank. For $2n+1$ qubits, the maximum CP rank is known to be on the order of $2^{2n-1}$\cite{Sumi2013_2tensors}. Based on numerical CP decomposition, we suspect $\rank(\R{2n+1}) = 3\cdot 2^n$, but the question remains open.

Beyond CP rank, we considered several multipartite entanglement measures on graph states based on bipartite measures. Surprisingly, these prove dichotomous: either 0 if the graph is disconnected, or a fixed value irrespective of graph structure beyond connectivity.

\section{Acknowledgements}
We would like to thank Yuchen Pang for helpful discussions. We are very grateful to Julio de Vicente for carefully pointing out how their work in Ref. \cite{Contreras-Tejada-2019a} implies Theorem \ref{thm:GME}.  The authors acknowledge support from the NSF Quantum Leap Challenge Institute for Hybrid Quantum Architectures and Networks (NSF Award 2016136). Linjian Ma and Edgar Solomonik were also supported via the US NSF RAISE/TAQS program, award number 1839204.

\bibliography{main}

\onecolumngrid
\pagebreak
\normalsize
\section*{Appendix}
\section*{Minimal Decomposition of Line States}
We find it informative to give a recursive minimal CP decomposition for line states $\L{n}$. Note that this can be used to construct the $3\times 2^{n}$ term CP decomposition for $\R{2n+1}$.
\begin{remark}\label{rmk:line}
    A minimal decomposition for any $\L{n}$ can be found via a simple recursive method. Define the following 2 qubit states:
    \begin{align}
        \ket{a} = \ket{0+},\ \ket{b} = \ket{1-},\ \ket{c} = \ket{0-},\ \ket{d} = \ket{1+}.
    \end{align}
    Any line state can be written as some sum over tensor products of these states. I.e. $\L{n} = \frac{1}{2^{\lfloor \frac{n}{2} \rfloor}}\sum_{\ket{\psi}\in\mathcal{L}_n}\ket{\psi}$, where all terms $\ket{\psi} $ are some string such as $\ket{aaa\ldots a}$, $\ket{acb\ldots d}$, etc. By $\mathcal{L}_n$ we denote the set of such strings in the decomposition of an $n$ qubit line state. In particular, the line state on $2n$ qubits takes the form
    \begin{equation}
        \L{2n} = \frac{1}{2^{n/2}}[(\ket{a}+\ket{b})\sum_{\substack{\ket{\psi}\in\mathcal{L}_{2(n-1)}\\\ket{\psi}_1 = \ket{0}}}\ket{\psi}+(\ket{c}+\ket{d})\sum_{\substack{\ket{\phi}\in\mathcal{L}_{2(n-1)}\\\ket{\phi}_1 = \ket{1}}}\ket{\phi}].
    \end{equation}
     Here by $\ket{\psi}_1 = \ket{0}$ we mean that the first party in $\ket{\psi}$ is in state $\ket{0}$. The line state on $2n+1$ qubits takes the form
     \begin{equation}
         \L{2n+1} = \ket{+}P_0^{(1)}\L{2n}+\ket{-}P_1^{(1)}\L{2n}.
     \end{equation}
    The base case is $\L{2} = \frac{1}{\sqrt{2}}(\ket{a}+\ket{b})$.
\end{remark}
\begin{proof}
    We can write a line state as 
    \begin{align}
        \L{n} = \frac{1}{2^{\lfloor \frac{n}{2}} \rfloor}\sum_{x\in\mathbb{F}_2^n} c(x)\ket{x},\ c(x) = \Pi_{i=0}^{n-2}(-1)^{x_ix_{i+1}}.
    \end{align}
    This follows from the action of $U^{(i,i+1)}$. Thus, we are looking to find a decomposition that contains all binary strings with equal weight, but with signs based off of the number of consecutive ones. It is clear that $\ket{a}$ and $\ket{b}$ satisfy these requirements. From inspection it is clear that $\ket{c}$ and $\ket{d}$ satisfy this property when following a party in the state $\ket{1}$. Further, $\ket{a}/\ket{b}$ combined and $\ket{c}/\ket{d}$ combined generate every binary string when expanded in the computational basis. When $n$ is even, we simply append new characters such that the sign properties are maintained. When $n$ is odd we do the same but with an addition of $\ket{\pm}$ chosen such that the sign property is maintained.
    
    The recursive structure above doubles the number of terms in the decomposition between $\L{2n}$ ($\L{2n+1}$) and $\L{2(n+1)}$ ($\L{2(n+1)+1}$). As the base case is rank 2, these are rank $2^{\lfloor \frac{n}{2} \rfloor}$ decompositions, which is optimal. 
\end{proof}

\section*{Upper Bound for 7 Qubit Ring}
Here we explicitly construct the rank $3\times 2^2=12$ decomposition for $\R{7}$. From remark~\ref{rmk:line} the 7 qubit line state can be written as:
\begin{align}
    \L{7} & = \frac{1}{2\sqrt{2}}\ket{+}(\ket{aaa}+\ket{acb}+\ket{cba}+\ket{cdb})+\frac{1}{2\sqrt{2}}\ket{-}(\ket{baa}+\ket{bcb}+\ket{dba}+\ket{ddb}),
\end{align}
\begin{align}
    P_0^{(1)}P_1^{(7)}\L{7}&= \frac{1}{4\sqrt{2}}\ket{0}[\ket{aa0}-\ket{ac1}+\ket{cb0}-\ket{cd1}+\ket{ba0}-\ket{bc1}+\ket{db0}-\ket{dd1}]\ket{1}.
\end{align}
Following the proof in the main text, we can find the desired decomposition from that of $\R{5}$, $\ket{a^{(4)}} = \frac{1}{2\sqrt{2}}\ket{0}(\ket{+0+}+\ket{-1-})$
and $\ket{b^{(4)}} = \frac{1}{2\sqrt{2}}\ket{0}(\ket{+0-}+\ket{-1+})$,
\begin{align}
    P_1^{(7)}\ket{\phi^{(7)}}=&\frac{1}{\sqrt{2}}[\ket{a^{(4)}}\ket{0-}+\ket{b^{(4)}}\ket{1+}]\ket{1}\nonumber\\ = & \frac{1}{4}\ket{0}[(\ket{+0+}+\ket{-1-})\ket{0-}+(\ket{+0-}+\ket{-1+})\ket{1+}]\ket{1}.
\end{align}
It can be verified that these are the same states via expanding into the computational basis.
As $\R{7} = U^{(1,7)}\L{7} = (\sigma_z^{(2n+1)}+2P_0^{(1)}\otimes P_1^{(7)})\L{7}$, we can thus write $\R{7}$ in the following 12 term decomposition
\begin{align}
    \R{7} = & \frac{1}{2\sqrt{2}}\ket{+}(\ket{0+0+0+}+\ket{0+0-1-}+\ket{0-1-0+}+\ket{0-1+1-})+\nonumber\\
    &\frac{1}{2\sqrt{2}}\ket{-}(\ket{1-0+0+}+\ket{1-0-1-}+\ket{1+1-0+}+\ket{1+1+1-})+\nonumber\\
    &\frac{1}{2}\ket{0}[(\ket{+0+}+\ket{-1-})\ket{0-}+(\ket{+0-}+\ket{-1+})\ket{1+}]\ket{1}.
\end{align}

\section*{Calculation of GME Entanglement for Graph States}

To show Theorem \ref{thm:GME}, we first observe the following.

\begin{corollary}\label{thm:qubit_rdm}
    The reduced density matrix for any individual qubit party i that corresponding to vertex v is 
    \begin{align}
        \rho_i = \begin{cases}
            \frac{1}{2}\mathbb{I} & \delta(v) > 0\\
            \ket{+}\bra{+} & \delta(v) = 0
        \end{cases},
    \end{align}
    where $\delta(v)$ is the degree of vertex $v$. 
\end{corollary}
\begin{proof}
    This follows readily from Lem.~\ref{thm:rdm}. If $v$ is not an isolated vertex, there is at least one non-zero value in $\Gamma_{A\overline{A}}$, where $A=\{v\}$, and thus $\rank(\Gamma_{A\overline{A}})=1$. If $v$ is an isolated vertex, then $\rank(\Gamma_{A\overline{A}})=0$.
\end{proof}



With this corollary and Lem.~\ref{thm:rdm}, we now show that the measures previously introduced are either 0 or a fixed constant based on if the graph is connected.

\begin{theorem}\label{thm:gme_conc}
    \begin{align}\cgme{\ket{G}}=\begin{cases}
    0 & \text{if G is a disconnected graph} \\
    1 & \text{otherwise}
    \end{cases}.\end{align}
\end{theorem}
\begin{proof}
    From Lem.~\ref{thm:rdm} we know that any reduced density matrix is maximally mixed on a certain subspace of dimension $2^d = 2^{\rank(\Gamma_{A\overline{A}})}$. By finding $\max_A\Tr[\rho_A^2]$ we minimize GME-Concurrence. The purity of a $k$-dimensional maximally mixed state is $\frac{1}{k}$. If there is a disconnected component $A$, $\rank(\Gamma_{A\overline{A}})=0$ and $\cgme{\ket{G}}=0$. Otherwise, the maximal purity is $\frac{1}{2}$, which is achieved by considering any single vertex. Thus, $\cgme{\ket{G}}=1$.
\end{proof}

\begin{theorem}\label{thm:gme_neg}
    \begin{align}\ngme{\ket{G}}=\begin{cases}
    0 & \text{if G is a disconnected graph} \\
    \frac{1}{2} & \text{otherwise}
    \end{cases}.\end{align}
\end{theorem}
\begin{proof}
    Before continuing, we note that negativity can be equivalently written as a summation of the absolute value of the negative eigenvalues of the partial transpose.
    
    \begin{align}
        \mathcal{N}(\rho^{AB}) & = \frac{1}{2}(\norm{\rho^{T_A}}_1-1)= \sum_{\lambda < 0}|\lambda|,
    \end{align}
    where $\lambda$ are the eigenvalues of $\rho^{T_A}$.
    
    Next, we use Lem.~\ref{thm:rdm} to write $\ket{G}$ in a Schmit decomposition $\ket{G}=2^{-d/2}\sum_{i=1}^{2^d}\ket{u_i}\ket{v_i}$, where $d=\rank(\Gamma_{A:\overline{A}})$. Thus, the partial transpose with respect to $\overline{A}$ is
    \begin{align}
        \mathbb{I}\otimes T(\dya{G}) = & 2^{-d}\sum_{i,j=1}^{2^d}\ket{u_i}\bra{u_j}\otimes\ket{v_j}\bra{v_i}.
    \end{align}
    This has negative eigenvalue $-2^{-d}$ with multiplicity $\binom{2^d}{2}$, corresponding to eigenvectors $\frac{1}{\sqrt{2}}(\ket{u_i}\ket{v_j}-\ket{u_j}\ket{v_i})$. Thus, the negativity according to partition $A|\overline{A}$ is
    \begin{align}
        \mathcal{N}(\rho^{A\overline{A}}) & = \binom{2^d}{2}2^{-d}
         = \frac{1}{2}(2^d-1).
    \end{align}
    Clearly this is increasing with $d$. Thus, the GME negativity will be minimized by a partition with the the smallest $\rank(\Gamma_{A\overline{A}})$. If $G$ is disconnected, there exists a partition $A$ such that $\rank(\Gamma_{A\overline{A}})=0$. Otherwise, $d=1$ for any partition into a single vertex, for which $\mathcal{N}(\rho^{A\overline{A}}) = \frac{1}{2}$.
\end{proof}

\begin{theorem}\label{thm:gme_geo}
    \begin{align}\ggme{\ket{G}} = \begin{cases}
    0 & \text{if G is a disconnected graph} \\
    \frac{1}{2} & \text{otherwise}
\end{cases}.\end{align}
\end{theorem}
\begin{proof}
    If $G$ is disconnected there is a partition with Schmidt coefficient 1 (Cor.~\ref{thm:qubit_rdm}). Otherwise, the largest Schmidt coefficient possible is always $\frac{1}{2}$ via Lem.~\ref{thm:rdm}.
\end{proof}

\end{document}